%% file: main-arxiv-V3.tex
\documentclass[11pt]{article}

\usepackage[top=1.0in, bottom=1.0in, left=1in, right=1in]{geometry}
\usepackage[T1]{fontenc}
\usepackage[tt=false, type1=true]{libertine}
\usepackage[varqu]{zi4}
\usepackage[libertine]{newtxmath} % Keep only newtxmath for math font
\usepackage{nicefrac}
\usepackage[utf8]{inputenc} % allow utf-8 input
\usepackage[T1]{fontenc}    % use 8-bit T1 fonts
\usepackage[hidelinks]{hyperref}       % hyperlinks
\usepackage{url}            % simple URL typesetting
\usepackage{booktabs}       % professional-quality tables
\usepackage{amsfonts}       % blackboard math symbols
\usepackage{microtype}      % microtypography
\usepackage{xcolor}         % colors
\usepackage{amsmath}

\usepackage{amsthm}
\usepackage{graphicx}
\usepackage[font=small]{caption}
\usepackage{subcaption}
\usepackage{bm}
\usepackage[ruled,noend]{algorithm2e}
\usepackage{algpseudocode} 
\usepackage{varwidth}
\usepackage{placeins}
\usepackage{mathtools}      % For enhanced math tools
\usepackage{booktabs}       % For tables
\usepackage{verbatim}
\usepackage{enumitem}
\setlist{itemsep=2pt}
\usepackage{scalerel}
\usepackage{cleveref}
\usepackage{nicefrac}
\usepackage{nicematrix}
\usepackage[numbers,sort&compress]{natbib}

\newcount\Comments
\usepackage{color}
\definecolor{darkgreen}{rgb}{0,0.5,0}
\newcommand{\kibitz}[2]{\ifnum\Comments=1{\color{#1}{#2}}\fi}

\newcommand{\bfpar}[1]{%
    \vspace{2pt}%
    \noindent%
    \textbf{#1:}%
}

\usepackage{color-edits}
\addauthor{et}{orange}
\addauthor{yk}{darkgreen}
\addauthor{jh}{blue}

\newcommand{\ignore}[1]{}
\newtheorem{theorem}{Theorem}[section]
\newtheorem*{theorem*}{Theorem} %unnumbered theorem
\newtheorem{lemma}[theorem]{Lemma}
\newtheorem{corollary}[theorem]{Corollary}
\newtheorem{proposition}[theorem]{Proposition}
\newtheorem{observation}[theorem]{Observation}

\newtheorem{claim}{Claim}
\newtheorem{definition}[theorem]{Definition}

\newtheorem*{example*}{Example}
\newtheorem*{remark}{Remark}

\algtext*{EndFor}% Remove "end for" text
\algtext*{EndIf}% Remove "end if" text

\DeclareMathOperator{\E}{\mathbb{E}}
\title{Games with Payments between Learning Agents}

\author{%
  Yoav Kolumbus \\
  Cornell University \\
  \texttt{yoav.kolumbus@cornell.edu} 
  \and
  Joseph Y. Halpern \\
  Cornell University \\
  \texttt{halpern@cs.cornell.edu} 
  \and
  \'Eva Tardos \\
  Cornell University \\
  \texttt{eva.tardos@cornell.edu} \\
}
\date{}

\begin{document}
% \sloppy
\maketitle

\begin{abstract}
\input{abstract}

\end{abstract}
%%%%%%%%%%%%%%%%%%%%%%%%%%%%%%%

% Paper content

%%%%%%%%%%%%%%%%%%%%%%%%%%%%%%%
\input{intro}

%%%%%%%%%%%%%%%%%%%%%%%%%%%%%%%
\input{model}

%%%%%%%%%%%%%%%%%%%%%%%%%%%%%%%
\input{further_related_work}

%%%%%%%%%%%%%%%%%%%%%%%%%%%%%%%
\input{auctions}

%%%%%%%%%%%%%%%%%%%%%%%%%%%%%%%
\input{single_player}

%%%%%%%%%%%%%%%%%%%%%%%%%%%%%%%
\input{two_players}

%%%%%%%%%%%%%%%%%%%%%%%%%%%%%%%
\input{conclusion}

%%%%%%%%%%%%%%%%%%%%%%%%%%%%%%%%

\section*{Acknowledgments} 
We thank Noam Nisan for his valuable feedback on this study. 
We also thank anonymous referees for helpful comments on an earlier version of the paper, and for useful feedback on Proposition~\ref{thm:iterated-dominant-strategies}.  
% Eva:
\'Eva Tardos was supported in part by 
AFOSR grant FA9550-23-1-0410, AFOSR grant FA9550-231-0068.    
% Joe:
Joe Halpern was supported in part by 
AFOSR grant FA23862114029, MURI
grant W911NF-19-1-0217, ARO grant
W911NF-22-1-0061, NSF grant FMitF-2319186, and a grant from the
Cooperative AI Foundation. 

\bibliographystyle{abbrvnat} % 
\bibliography{auction-with-payments-refs} % 

\appendix
\section*{APPENDIX}
%%%%%%%%%%%%%%%%%%%%%%%%%%%%%%%%
\input{appendix_PD_example}

%%%%%%%%%%%%%%%%%%%%%%%%%%%%%%%%
\input{appendix_FP_proof}

%%%%%%%%%%%%%%%%%%%%%%%%%%%%%%%%
\input{appendix_FP_n_player_NE}

%%%%%%%%%%%%%%%%%%%%%%%%%%%%%%
\input{appendix_standard_definitions}

%%%%%%%%%%%%%%%%%%%%%%%%%%%%%%%%

%%%%%%%%%%%%%%%%%
\end{document}

%% file: abstract.tex
In repeated games, such as auctions, players rely on autonomous learning agents to choose their actions. We study settings in which players have their agents make monetary transfers to other agents during play at their own expense, in order to influence learning dynamics in their favor. Our goal is to understand when players have incentives to use such payments, how payments between agents affect learning outcomes, and what the resulting implications are for welfare and its distribution. We propose a simple game-theoretic model to capture the incentive structure of such scenarios. We find that, quite generally, abstaining from payments is not robust to strategic deviations by users of learning agents: self-interested players benefit from having their agents make payments to other learners. In a broad class of games, such endogenous payments between learning agents lead to higher welfare for all players. In first- and second-price auctions, equilibria of the induced ``payment-policy game'' lead to highly collusive learning outcomes, with low or vanishing revenue for the auctioneer. These results highlight a fundamental challenge for mechanism design, as well as for regulatory policies, in environments where learning agents may interact in the digital ecosystem beyond a mechanism's boundaries.

%% file: intro.tex
\section{Introduction}\label{sec:intro}
Autonomous learning agents have become widespread on online platforms, playing an increasingly pivotal role in markets and economic ecosystems. A prominent example is the multi-billion-dollar online ad auction industry, estimated at over 1\% of U.S. GDP \cite{marto2024rise}, led by major corporations like  Google, Meta, Amazon, and Microsoft. Due in part to the high frequency of these auctions, automated bidding has emerged as the dominant approach, with bidding traffic managed by various automated agents, provided either directly by the platforms themselves, or by
third parties. Typically, these agents get some high-level instructions from their users about their objectives and allowed action space, and then autonomously interact with other agents in long sequences of repeated games (which could include thousands or millions of auctions per hour), using some learning algorithm to optimize the long-term payoffs for their users.

We are interested in repeated interactions of this sort where players use learning agents to engage in repeated games on their behalf. We take as a point of departure the observation that, in this setting, users may wish to allow their agents to make payments to other agents in order to influence the course of the game dynamics in their favor. 
In this study, we model and analyze the effect of considering such payments between agents on the long-term outcomes of the interaction.

More broadly, the idea that payments and other types of financial interactions outside a mechanism can affect behavior within the mechanism has received considerable attention across a wide range of domains. In blockchain fee markets, for example, this concern has become central to the evolving design of these markets (see \cite{chung2024collusion,roughgarden2021transaction} on off-chain-agreement robust design and \cite{bahrani2023transaction,bahrani2024transaction,daian2020flash} on Maximal Extractable Value (MEV), where algorithmic agents strategically use payments to manipulate market outcomes). Related considerations also arise in fair division \cite{moulin1992application,moulin2004fair}. We defer a more extended discussion of related work to Section~\ref{sec:related-work}. 

In classic strategic settings, such as institutional auctions or voting, collusion and side payments are well-recognized concerns (see, e.g., \cite{hendricks1989collusion,kawai2022detecting} and \cite{harstad2007harmonization}) and are primarily addressed through legal and regulatory frameworks. By contrast, interactions between algorithmic agents on online platforms are not yet well understood and still lack clear regulatory boundaries and effective enforcement tools (see also the references in Section~\ref{sec:conclusion}).

The possibility that learning agents could transfer payments among themselves during their game dynamics raises several basic questions. First, \emph{when do players have incentives to let their agents use payments?} Second, \emph{how do payments between learning agents affect their dynamics?} And third, \emph{what are the long-term implications of payments among agents for  the players' utilities and  social welfare?  (And in the case of auctions, what happens to the seller's revenue?)}  

% \vspace{3pt}
As in much of the prior literature (e.g., \cite{blum2007learning,daskalakis2016learning,dworczak2020mechanism,foster2016learning,kolumbus2022auctions,zinkevich2003online}), we model agents as no-regret learners; that is, they use learning strategies that 
result in outcomes satisfying the no-regret condition in the long term.
Unfortunately, as is well known, in many games 
no-regret agents can end up playing strategies that yield low utility compared to cooperative outcomes that could potentially have been obtained in the game. 
E.g., in the prisoner’s dilemma, they end up defecting, since that is the dominant strategy. This has led to a great deal of interest in understanding what natural mechanisms may lead agents to get, in some sense, better outcomes. 

Notably, unlike the literature on mechanism design, where a platform or some external entity provides incentives or recommendations aiming to improve welfare \cite{babaioff2022optimal,balcan2009improved,monderer-tenenholtz-2004,SteeringNoRegret2023}, in our case, there is no external entity. All payments are made at the expense of the players, who are solely interested in maximizing their own individual payoffs.

%%%%%%%%%%%%%%%%%%%%%%
Before giving an example, let us sketch the general structure of the model we have in mind. As mentioned, there is an underlying game, and the model of what we call a \emph{payment policy game} (or just a \emph{payment game}) proceeds in two phases, where each \emph{player} $i$ uses a learning \emph{agent} $i$. In the first phase, each player $i$ chooses a \emph{payment policy} that determines, for each other agent $j$, how much agent $i$ will pay agent $j$ at each round of the agents' dynamics as a function of the outcome (i.e., the joint action profile performed by the agents) in the underlying game in that round. In the second phase, the agents play the underlying game repeatedly, choosing their actions according to their learning algorithm (taking the payments into account). See Section \ref{sec:model} for the full formal description of the model. 
%%%%%%%%%%%%%%%%%%%%%%%%%%%%

The purpose of using this stylized metagame framework (see references in Section \ref{sec:related-work}) is to isolate the basic incentive structure that arises once learning agents are allowed to make endogenous payments to one another, and to study whether standard conclusions about learning dynamics are robust to this extension. We study equilibria of the resulting payment-policy game under standard no-regret learning dynamics, and show that such equilibria typically involve payments that substantially alter the outcomes of the underlying game.\footnote{A natural question is how equilibria of the metagame might be reached, or which equilibria are selected. One could, for example, model an additional layer of learning over the space of payment policies and analyze specific learning dynamics. We do not pursue this direction, and instead focus on structural properties of the resulting metagames.} In particular, we show that, very generally, when agents achieve low regret, refraining from payments is not a stable behavior. As a result, even weak forms of optimization over payment policies (such as better responses) would induce agents to trade outside the platform.

%%%%%%%%%%%%%%%%%%%%%%%%%%%%
\bfpar{Warm-up example} With this background, consider the variant of the prisoner's dilemma game 
shown in Figure \ref{fig:PD-game-a}. 
As mentioned above, standard analysis shows that when both players use no-regret learning agents, each agent will quickly learn to
play its dominant strategy, which is to defect (i.e., play strategy $D$); the dynamics will converge to the Nash equilibrium of the stage game where both agents always defect and each player gets an expected payoff of $\nicefrac{1}{3}$.

\begin{figure}[t!]
\centering
    \begin{NiceTabular}{cccc}[cell-space-limits=3pt]
         &     & \Block{1-2}{{\small Player $2$}} \\
         &     &  {\small$C$}     &  {\small$D$} \\
    \Block{2-1}{{\small Player $1$}} 
         &  {\small$C$} & \Block[hvlines]{2-2}{}
               $\nicefrac{2}{3}, \nicefrac{2}{3}$ & {\small$0, 1$} \\
         &  {\small$D$} &  {\small$2, 0$} & $\nicefrac{1}{3}, \nicefrac{1}{3}$ 
    \end{NiceTabular}
\caption{A prisoner's dilemma game.}
\end{figure}
\label{fig:PD-game-a}
 
Now we ask what can player $1$ achieve if she
allows her agent to make payments to the other agent during the game dynamic? 
Consider the payment policy where player $1$'s agent
pays player $2$'s agent  $\nicefrac{1}{3} + \epsilon$
if player $2$'s agent cooperates, for some $\epsilon > 0$, and pays
nothing if player 2's agent defects. 
It is easy to see that in the game with these payments agent $2$ has a dominant strategy, and agent $1$ has a strict best response, leading to the strategy profile $(D,C)$ as the 
unique equilibrium, and so the ensuing dynamics will converge to that outcome 
(see Appendix \ref{sec:appendix-PD-game}), 
with payoffs $(2-\nicefrac{1}{3} - \epsilon, \nicefrac{1}{3} + \epsilon)$ for the two players.

It turns out that this outcome can be obtained as an ($\epsilon$-)equilibrium\footnote{In an $\epsilon$-equilibrium,  no player can increase their payoff by more than $\epsilon$ by altering their strategy.} of the payment game. 
There is also a second similar $\epsilon$-equilibrium of the payment game, obtained by switching the roles of players $1$ and $2$, but due to the asymmetry in the payoffs of the underlying game, the social welfare in this second equilibrium is $1$, not $2$. Importantly, however, both equilibria are Pareto improvements over the Nash equilibrium of the underlying game, and the players' welfare significantly increases compared to the game without payments (note, though, that welfare is not necessarily optimal).

We show, using a similar analysis, that we get these equilibria and improvement in social welfare in general prisoner's dilemma games whenever the maximum welfare gap in the game is large (above a factor of $2$). 
Moreover, while in the game without payments the \emph{price of anarchy} ($PoA$) and \emph{price of stability} ($PoS$)---the ratio between the social welfare obtained from the worst-case (respectively, best-case) Nash equilibrium and the optimal social welfare---are unbounded, with payments, both $PoS$ and $PoA$ are bounded by a factor of $2$ in the symmetric case, and in the asymmetric case $PoS$ is bounded, but $PoA$ is not, due to the welfare gap between the two equilibria of the payment game. See Appendix \ref{sec:appendix-PD-game} for further details.

The point of the example is not the specific payoffs, but the fact that best-response payments change which outcomes are consistent with regret minimization. 
Although prisoner's dilemma is very simple (both players have
strictly dominant strategies), the observations we made regarding this example apply more broadly. In a wide range of games, \emph{selfish players benefit from using payments between their agents}, and the use of payments
can often enable players to reach more cooperative outcomes, with
higher social welfare.  

\vspace{3pt}
\bfpar{Our contribution} 
We study the incentives that players have to use payments with their learning agents and the potential impact of such payments among learners on the joint  dynamics and outcomes.

From a high-level perspective, our results highlight a challenge for auction design and market design more broadly: we show that users of learning agents quite generally have inherent incentives to allow their agents to make payments during their interactions. When these incentives are coupled with the right technology---particularly, sophisticated AI agents and flexible transaction media, such as those on blockchain platforms---one can easily imagine an outlook  where agents trade among themselves ``under the hood,'' and the associated markets change their behavior, equilibria, and outcomes. Our analysis takes a step toward understanding these interactions. The  
results demonstrate that these changes can be very significant, underlining, on the one hand, a potential improvement in efficiency, but on the other hand, a risk, and the need to better understand the impact on concrete systems, to 
be able to design them accordingly. 

We propose a simple model of these interactions for general games in Section \ref{sec:model}. As our main case study, we analyze first- and second-price auctions in Section \ref{sec:auctions}. In the following sections, we study general properties of payment games: Section~\ref{sec:single-player-manipulations} focuses on the instability of not using payments, by analyzing the incentives players have to initiate unilateral payments, and Section~\ref{sec:two-player-games} focuses on the special case of two-player games, where stronger cooperative learning outcomes emerge due to payments. Our results in all these settings provide comparative statics between the outcomes of learning dynamics without payments and those obtained with learning agents in equilibria of the payment game between their users. 
% \vspace{5pt}
% \noindent
% \bfpar{Main Theorems}
To conclude the introduction, we next give a brief informal summary of the main theorems. 

In second-price auctions, Theorem \ref{thm:second-price-zero-revenue} shows that players have incentives to provide payment policies to their agents (i.e., using learning agents in the auction without using payments is not a Nash equilibrium), and that in an equilibrium of the payment game, in the long run, the players may reach full collusion, where they capture almost the entire welfare, leaving the auctioneer with vanishing revenue as $T \rightarrow \infty$. This result applies for any number of players with generic (i.e., not exactly equal) valuations. 

In first-price auctions, Theorem \ref{thm:first-price-cooperative-outcomes} shows that here as well, players have incentives to use payments. 
So using automated bidders without using payments is not an equilibrium. 
This result applies for any number of players using mean-based no-regret algorithms (see Appendix \ref{sec:appendix-FP-auction-n-player-NE} for further details). For the special case of two players, the theorem also shows that equilibria of the payment game lead to low (but in this case, still positive) revenue.

In Section~\ref{sec:single-player-manipulations}, we show that under broad conditions, players have incentives to deviate from zero payments and use payments with their learning agents, also beyond auctions. Theorems~\ref{thm:single-player-optimal-welfare} and~\ref{thm:single-player-stability} show that in a wide range of finite games,\footnote{Finite games are games with bounded payoffs and, unlike auctions, finite action sets.} there is at least one player who has an incentive to use payments, and we characterize cases where a payment policy applied by a single agent can drive the dynamics toward the optimal-welfare outcome while also improving the payoff for its user.

Theorems \ref{thm:two-players-Pareto-improvement} and \ref{thm:two-players-DS-games} deal with the case of two-player games. Bilateral interactions have special  properties in payment games, intuitively, stemming from the fact that any payment must be desirable to both players, as there are no other externalities on the associated agents from interactions with third parties. Theorem \ref{thm:two-players-Pareto-improvement} shows that in two-player games, payments always result in Pareto improvements compared to the game without payments. Furthermore, Theorem \ref{thm:two-players-DS-games} demonstrates that in a broad class of games --- extending even to games where  players have strictly dominant strategies --- at least one player will strictly benefit from using payments. Consequently, not using payments is not a stable behavior.

Taken together, these results suggest a clear big picture: scenarios in which automated learning agents, each optimizing its own objective, do not use payments are strategically unstable. Very broadly, learning agents (and hence their users) benefit from using payments at their own expense to influence learning dynamics. Equilibria of payment policies, or even unilateral deviations from zero payments, can substantially impact learning dynamics. These incentives to introduce payments into learning dynamics arise naturally, and the resulting payoffs, revenue, and social welfare can differ markedly from outcomes induced by learning dynamics without payments.

%% file: model.tex
\section{Model}\label{sec:model}
We consider scenarios where we have players using 
learning agents in some repeated game, and study the setting where players can augment their agents by allowing them to make payments to other agents based on 
the actions of these agents during the game dynamics.  
The learning agents are assumed to satisfy the regret-minimization property (see Appendix \ref{sec:appendix-definitions} for standard definitions), 
but not restricted to using any particular algorithm, and the players themselves are interested in their own
long-term payoffs. 
As our analysis relies only on the long term learning outcomes of the agents in the limit $T \rightarrow \infty$, where $T$ is the number of game rounds, our results apply regardless of the specific regret bounds or details of the algorithms used, assuming that the agents achieve regret sublinear in $T$. The model we consider consists of the following components. 
\begin{itemize}[leftmargin=*, itemsep=3pt]
    \item \textbf{The underlying game}: There is an underlying game $\Gamma = \{[n],S,\{u_i\}_{i=1}^n\}$, where $[n]$ is the set of players, $S$ is the space of joint actions, and $u_i:S \rightarrow [0,1]$, $i = 1, \dots, n$ are the utility functions.      
      
    \item \textbf{Agents}: Every player uses a no-regret agent to play the game on their behalf for $T$ rounds. Denote the set of these agents by $A$, where $A_i \in A$ is the agent of player $i\in [n]$. 
    Agents choose actions using their no-regret algorithms. The action of agent $A_i$ at time $t \in [T]$ is denoted $s_i^t$ and the action profile of all the agents is denoted $s^t$.  
      
    \item \textbf{Payments}: In addition to choosing actions, the agents make payments to other agents that depend on the actions $s^t$ chosen by the agents, according to payment policies defined by the associated players (the agents' users), as specified next.
            
    \item \textbf{Payment policies}: Each player $i$ chooses a policy for her own agent that determines, for each action profile $s \in S$ and player $j \in[n]\setminus\{i\}$, the payment $p_{ij}(s) \in \mathbb{R}_+$ from agent $i$ to agent $j$ when the outcome in the most recent step is $s$. We assume that payments are bounded: $p_{ij}(s) \in [0,M]$, where $M > n$ is an arbitrary large constant.  
    
    \item \textbf{Agent utilities}: 
    In every step $t$, the agents play an action profile $s^t \in S$ and observe the realized outcome and payments. The agent-utility of agent $i$ at time $t$ 
    is $v_i^t = u_i(s^t) + \sum_{j \neq i} \big(p_{ji}(s^t) - p_{ij}(s^t)\big)$, that is, the utility from the underlying game plus the net payment.
      
    \item \textbf{Players' utilities}: The players' utilities are their long-term average payoffs: $U_i^T = \E[\frac{1}{T} \sum_{t=1}^T v_i^t]$; we are particularly interested in the limit $\lim_{T\rightarrow \infty}U_i^T$, denoted by $U_i$.      
    % \item 
    % Since in the analysis of long-term payoffs, finite-time costs vanish, 
    % we assume w.l.o.g. that 
    If a player is asymptotically indifferent as $T \rightarrow \infty$ between making a payment and not making the payment, the player prefers not to make the payment due to (vanishing) transient costs.   
\end{itemize}

Note that the payment policies remain fixed for periods sufficiently long that agents achieve no-regret outcomes. The goal of using this metagame abstraction is to isolate the underlying incentive structure of the resulting interaction induced by payments, and it can be viewed as a separation of time scales between changes in payment policies and play in the underlying game (e.g., bids in online auctions), similar in spirit to~\cite{aggarwal2024randomized,feng2024strategic,KolumbusNisan2021manipulate,kolumbus2022auctions}.

The model above defines a \emph{payment-policy game} (in short, a \emph{payment game}) between users of learning agents, 
where an instance of a payment game is specified by $\mathcal{G} = \{\Gamma, A, T \}$,  
and the players' actions in the payment game are to choose their payment policies $p_{ij}(s)$ for each $s \in S$, as defined above. 
To understand the basic properties of
payment games and how the incentives of players change compared
to games without payments, our analysis of equilibria and potential
deviations focuses on one-shot, full-information payment-policy games.
A fundamental question that arises is: \emph{when would players have incentives to use payment policies?}
Conversely, what are the conditions for a game to \emph{``stable''}, in the sense that players do not have incentives to use such manipulations?

\vspace{5pt}
\begin{definition}\label{def:stable-games}
    A game $\Gamma$ is \emph{``stable for a set $A$ of learning agents''}
    if zero payments are an equilibrium of the payment-policy game $\mathcal{G}$ associated with $\Gamma,A$, for sufficiently large $T$. A game $\Gamma$ is \emph{``stable''} if it is stable for any set of regret-minimizing
    agents for the players.  
\end{definition}

%% file: further_related_work.tex
\section{Further Related Work}\label{sec:related-work}
\textbf{Learning in games:} This paper follows a long research tradition on learning and dynamics in games, from early work of Brown, Robinson, Blackwell, and Hannan in the 1950s \cite{blackwell1956analog, brown1951iterative,hannan1957lapproximation,robinson1951iterative}, through  seminal work in the following decades \cite{foster1997calibrated,fudenberg1995consistency,hart2000simple} and continuously since \cite{abel2025learning,easley2025markets,ROI-no-regret-autobidding-Aggarwal-Fikioris-Zhao-2024,blum2007external,daskalakis2021near,daskalakis2018last,fikioris2023liquid,
% hartline2015no,
kalai2005efficient,kolumbus2024asynchronous,kolumbus2022auctions,KolumbusN22,kumar2024strategically,milionis2023impossibility,syrgkanis2015fast}. See \cite{cesa2006prediction,hart2013simple} for an overview of the foundations of the field. The notion of regret has been central throughout this work as a tool to define learning outcomes and objectives and design algorithms to achieve them.

\bfpar{Regret minimization in auctions} No-regret learning has been widely studied in auctions, with two prominent lines of work focusing on
the price of anarchy in various auction formats with regret-minimizing bidders \cite{blum2008regret,caragiannis2015bounding,daskalakis2016learning,
% hartline2015no,
roughgarden2012price,roughgarden2017price,syrgkanis2013composable} and on the 
% econometric 
problem of inferring bidder preferences, where such bidder behavior is either assumed or empirically verified  
\cite{gentry2018structural,nekipelov2015econometrics,noti2017empirical,nisan2017quantal}. Other work has studied automated bidding algorithms for auctions with more complex features, such as budget constraints \cite{balseiro2019learning,deng2021autoBidding,fikioris2023liquid,fikioris2024learning,lucier2024autobidders}, or learning reserve prices \cite{cesa2014regret,mohri2014optimal,roughgarden2019minimizing}. In  \cite{alaei2019response,noti2021bid}, regret minimization is used as a prediction model for bidder behavior. Importantly, \cite{nekipelov2015econometrics,noti2021bid} provide empirical evidence from large auction datasets showing that bidder behavior in practice is, by and large, consistent with no-regret learning outcomes. 
Learning dynamics in first- and second-price auctions and their convergence properties have been studied in a broad range of work  \cite{bichler2023convergence,borgs2007GFPdynamics,daskalakis2016learning,deng2022nash,feng2020convergence,fikioris2023liquid,fikioris2024learning,kolumbus2022auctions,KolumbusN22,lucier2024autobidders}. 

\bfpar{Metagames}
A recent line of work 
studies metagames between users of learning agents \cite{kolumbus2022auctions,KolumbusN22}, in which players set parameters for their agents, which then repeatedly interact in some underlying game (e.g., an auction, a market, or an abstract normal-form game). The utilities in the metagame are determined by the long-term outcomes of the agents' dynamics. These papers focus on analyzing the average outcomes of regret-minimizing agents and their implications on users' incentives to misreport their preferences to their own agents. Further research has explored metagames in randomized auctions \cite{aggarwal2024randomized} as well as strategic budget selection in auctions \cite{alimohammadi2023incentive,feng2024strategic,mehta2023auctions} and its implications to welfare and revenue.

Our analysis is closely related to the metagame framework but focuses on a distinct concern and type of interaction that has not been analyzed for autonomous learners: endogenous payments between agents.
We show that payments have a markedly different effect on the agents' dynamics, leading to different analyses and results in auctions and other games compared to the strategic parameter reports studied in prior work. For example, in both first- and second-price auctions, equilibrium bids and revenues are lower in our payment-based scenario than in prior metagame results \cite{kolumbus2022auctions,aggarwal2024randomized}.

\bfpar{Optimization against learning agents}
% and optimal commitments:} 
A different line of research on interactions with learners considers scenarios where an optimizer interacts with a learning agent, such as an auctioneer optimizing auction rules against a learning buyer \cite{braverman2018selling}. This has been extended to studying cases with multiple buyers \cite{cai2023selling}, general two-player games \cite{arunachaleswaran2024pareto,deng2019strategizing}, Bayesian games \cite{MansourMohriSchneiderSivan2022strategizing}, and repeated contracting \cite{guruganesh2024contracting}.  
While that line of work shares the premise of players committing to using learning algorithms, its focus is different: it does not study payments between the agents, and furthermore, it addresses an {\em optimization problem} where a single player solves for the best response to a fixed learning algorithm, while we study a {\em game} (where none of the players is a Stackelberg leader) and the equilibrium outcomes that arise from the long-term dynamics of multiple learning agents.      

\bfpar{Algorithmic Collusion} A growing body of work studies algorithmic collusion and its relationship to no-regret learning in repeated games. Hartline et al.~\cite{hartline2024regulation} relate notions of non-collusion to swap-regret guarantees, while Arunachaleswaran et al.~\cite{arunachaleswaran2025algorithmic} show that the relationship between no-regret learning and collusion is subtle: some subclasses of no-regret algorithms admit collusive or monopolistic outcomes, while others preclude them. In closely related work, Arunachaleswaran et al.~\cite{arunachaleswaran2025swap} study robustness properties of learning dynamics, characterizing notions of non-manipulability in general games. Beyond no-regret analysis,~\cite{calvano2020artificial} study collusion in pricing games with simple Q-learners, and~\cite{fish2024algorithmic} demonstrate algorithmic collusion in extensive experiments with LLM agents. 
The present work is complementary and orthogonal to this literature. Monetary transfers among agents fundamentally change the payoff structure, and thus affect learning dynamics through a different channel than those considered in the above works; as a result, existing notions of non-manipulability or collusion robustness do not suffice.

\bfpar{Cooperation in repeated games} The question of how distributed and self-interested players can achieve cooperative or efficient outcomes is a broad topic, studied across various, often separate, fields. A classic approach addresses this challenge through mechanism design \cite{myerson1989mechanism,nisan1999algorithmic} and implementation theory \cite{jackson2001crash,maskin2002implementation}. Somewhat closer to our work are papers that study how an external mechanism designer can minimize the exogenous payments she needs to make to the players in order to implement particular outcomes in the game. 
See, e.g., \cite{babaioff2022optimal,monderer-tenenholtz-2004}, and \cite{SteeringNoRegret2023}, where a designer adds exogenous payments to learning agents in order to steer their convergence to a desired outcome. By contrast, we study the impact of payments that automated agents could make among themselves; our analysis does not consider a mechanism-design problem and there are no external funds or a directing hand. 
An additional strand of work from the reinforcement learning (RL) and artificial intelligence literature studies how distributed RL algorithms can jointly reach cooperative outcomes in sequential social dilemmas \cite{austerweil2016other,eccles2019learning,jaques2019social,leibo2017multi,yang2020learning}. A third line of work considers the notion of program equilibrium 
\cite{lavictoire2014program,
% oesterheld2019robust,
tennenholtz2004program}, in which each player declares a program and each program can read the commitments made in other programs and condition actions in the game on those commitments. 
This model leads to a broad range of equilibria, including cooperative ones, but differs significantly from our setting.

In comparison to research on collaborative outcomes from the perspective of cooperative game theory \cite{moulin1995cooperative}, we note that while in our payment game the players do manage to share, to some extent, the welfare they obtain, the game is not cooperative. In particular, an equilibrium of the payment game is not necessarily an element of perhaps the most standard solution concept in cooperative games, the core \cite{gillies1959solutions}. This can be clearly seen, for instance, in our analysis of first-price auctions, where in equilibrium the bidders do not play an action profile in the core of their game (thinking of the auction mechanism as set, and the game as taking place only between the bidders).

\bfpar{Games with monetary transfers} 
The concept of using payments to incentivize desired behaviors in games is a well-established idea. Payments between players have been studied in the context of social choice \cite{harstad2007harmonization} and fair division \cite{moulin1992application,moulin2004fair}, and are a central concept in the broad literature on contract theory and principal-agent problems. For an introductory overview on this literature, we refer interested readers to \cite{bolton2004contract,duettingcontract,laffont1981theory}.

Most work in this area involves a ``principal'' who provides incentives to agents, thereby facing a contract-design or optimization problem rather than acting as a strategic player in the game. Notable papers on contracts more closely related include \cite{jackson2005endogenous}, which studies multi-lateral payment contracts in one-shot interactions; \cite{guruganesh2024contracting,zhu2022sample}, which study learning in repeated contracts, \cite{ben2024principal} in MDP contracting settings, and \cite{ramirez2023game}, which explores scenarios where a player extracts fees from other players by using binding contracts to alter the equilibrium of the game. 
Our model however, diverges significantly from the standard contract framework.
Instead of focusing on principal-agent dynamics or external incentives injected into the game, we analyze scenarios where players choose to augment their learning agents with the capacity to make payments. Our primary interest lies in when such payments are beneficial and how they influence the agents' learning dynamics and, consequently, the outcomes of the games they play. Other work on games with payments related to ours include \cite{dutta2019asynchronous} and \cite{kosenko5088495efficiency}, who study sequential games with transfers in which players make sequential commitments, responding rationally to the commitments made by previous players. They study conditions where the sequential commitment structure leads to a unique subgame perfect equilibrium with full welfare.  These works do not deal with learning agents and their setting and analysis differ significantly from ours.

%% file: auctions.tex
\section{Auctions with Payments between Learning Agents}\label{sec:auctions}

To gain a deeper understanding of the incentives of users in online platforms to use payments with their learning agents and the potential impact of such payments, we focus on one class of games that has already been widely analyzed: auctions with automated bidders. 
We consider the simplest model of first- and second-price auctions, 
where at each time $t \in [T]$ a single identical item  is sold. Each
player $i$ has a value $v_i$ for the item, and we index the players in
decreasing order of their values, $v_1 \geq v_2 \geq \dots \geq v_n$. At round $t$, each agent submits a bid $b_i^t$ and the
auction mechanism determines the identity of the winner and a price
$p^t$. The payoff for the winning agent from the auction is $v_i -
p^t$; the other agents get  payoff zero. 

In addition to payoffs from the auction, agents
get a payoff according to the payments between the agents (as defined in Section \ref{sec:model}). 
Utilities are additive over game rounds. In the first-price auction, the payment is equal to the highest bid; 
in the second-price auction, the payment is the second-highest bid. In both cases, the highest bidder wins, and we assume tie-breaking according to the index of the agent (as is done, e.g., in \cite{feldman2016correlated}), 
so that if $i < j$ and $i$ and $j$ both have the same top bid, then $i$ is taken to be the winner.  
All our analyses are similar for arbitrary tie-breaking rules, but breaking ties in favor of the highest-valued agent simplifies the presentation.

\subsection{Second-price auctions}
The second-price auction has been widely studied \cite{vickrey1961counterspeculation} (see Section \ref{sec:related-work} for further references). 
While the auction is known to be incentive-compatible (or ``truthful''), it is also known that truthful bidding forms only one out of many Nash equilibria of the stage game. In the context of learning dynamics, it has been shown that regret-minimizing agents in this auction do not converge to
the truthful equilibrium; instead, they reach some degree of tacit collusion \cite{kolumbus2022auctions} with lower revenue to the
auctioneer than that obtained under truthful bids.  

Specifically, simplifying slightly, second-price auctions have two types of
non-truthful equilibria. 
The first type is 
``overbidding equilibria,'' 
where the high-valued player bids 
anything above her value
and the low-value player 
bids anything below the high value. The second type is ``low-revenue equilibria,'' where the
high-valued player bids anything above the low value and the low-value
player bids anything below the low value.\footnote{Both types of equilibria arise for any number of players $n$; we describe the two-player case for ease of exposition.} 
The latter type is especially interesting for our setting since any mixture of such
equilibria is consistent with regret minimization, but yields high welfare to the players. 
This suggests that when 
learning agents exchange payments, 
the higher-valued agent may be able to influence the dynamics with the other agents to induce 
a better equilibrium for herself.

We show that, indeed, when agents ``trade outside of the mechanism'' in this way during their learning dynamics, this kind of collusion emerges from the agents' interaction in a strong way, 
potentially leading to zero revenue for the auctioneer. Notably, these outcomes can be supported by simple payment policies. 

\begin{theorem}\label{thm:second-price-zero-revenue}
    The single-item second-price auction with $n$ players, where the values of the  players with the two top values are $v_1 > v_2$, is not stable for any regret-minimizing agents for the players. Furthermore, the payment-policy game of this auction has an $\epsilon$-equilibrium where the players capture the full welfare, and the auctioneer gets zero revenue.
\end{theorem}

The idea of the proof is, first, to construct  a
unilateral payment policy for player $1$, parameterized by a single
parameter $\epsilon$ that determines the amount that agent $1$ pays to the other agents, such that this policy induces a dynamic where  player $1$'s (i.e., the high-valued player's) agent bids its true value and all the other learning agents bid zero. We then show that for every other player, using the zero-payment policy is in fact a best response, and that the (net) payoff for the high-valued player can be
arbitrarily close to the full welfare, so player $1$ cannot improve her payoff by more than $\epsilon$ for any $\epsilon > 0$. (Note that with the second-price payment rule, if agents were not restricted to using finite payments,  players would be locked into a game without an equilibrium, where each player tries to commit to making the highest bid and incentivize the others to bid zero.)

\begin{proof} (Theorem \ref{thm:second-price-zero-revenue}). 
Consider a repeated second-price auction with a single item sold in every step, and $n$ players with values $v_1 > v_2 \geq \dots \geq v_n$. Assume that the total payment that a player makes is bounded by some large constant $M >  2 n v_1$.     
Consider the following profile of payment policies. Every player $j > 1$ uses the policy of always paying zero. Choose $\epsilon$ with $0 < \epsilon < (v_1 - v_2)/n$. The payment policy for player $1$ is specified by the following conditions. (1) Whenever $b_1 = v_1$, make a payment $p_{1j}$ to every player $j > 1$ of $\epsilon/(n-1)$ if $b_j = 0$ and zero otherwise. (2) If agent $1$ bids $b_1 \neq v_1$, then she pays {\small $p_{1j} =  \frac{M}{n-1}$} to every other player $j > 1$. Note that the maximum payment according to these conditions is $M$, as required. We continue with the following claim.
    
\begin{claim}
    With these payments, bidding $v_1$ is a strictly dominant strategy for agent $1$.
\end{claim}
\begin{proof}
    To see this, fix a bid profile for the other agents $2,\dots,n$, let $b_0$ denote the maximum bid of the other agents, and let $p$ denote the total payment to the other agents due to condition (1). We have the following cases.

    \vspace{3pt}
    \noindent
    Case 1 $b_0 \leq v_1$: In this case, the bid $v_1$ gives agent $1$ a utility of $v_1 - b_0 - p$. Every other bid $b \neq v_1$ and $b \geq b_0$ gives a utility of $v_1 - b_0 - M$, which is strictly less. Every losing bid $b < b_0$ gives a utility of $- M$, which is also strictly less. 

    \vspace{3pt}
    \noindent
    Case 2 $b_0 > v_1$ (if agents overbid): Here every winning bid for agent $1$ gives a utility of $v_1 - (b_0 + M)$, every losing bid not equal to $v_1$ gives a utility of $- M$, and bidding exactly $v_1$ gives a utility of $-p$, which is strictly higher. Therefore, bidding $v_1$ is a strictly dominant strategy, and agent $1$ who is minimizing regret will learn to only use this strategy.
    \end{proof}
  
    Next, we claim that since agent $1$ bids only $v_1$, every other agent has a strict best response, which is to bid zero. This is clear, since for every agent $j > 1$, every winning bid gives negative utility, every losing bid $b \leq v_1$ and $b \neq 0$ gives zero utility, and bidding exactly zero gives positive utility. Thus, with these payment policies, due to the regret-minimization property, agent $1$ always bids $v_1$ and wins the auction, while the other agents bid zero and get a payment of $\epsilon/(n-1)$ each. The utility for player $1$ in the long term is $v_1 - \epsilon$, and the utility for every other player is $\epsilon/(n-1)$. 
       
\begin{claim}\label{thm:claim-second-price-BR-with-overbidding}
The policy of always paying zero is a best response for every player $j > 1$.
\end{claim}
\begin{proof}
For simplicity, and without loss of generality, we rescale utilities and payments and work in units where $v_1 = 1$. We begin by looking at auctions with $n > 2$. 
In every deviation by some agent $j>1$, in order to ensure that a bid $b < v_1$ is not dominated for agent $1$, agent $j$ must pay agent $1$ an amount of at least $M - v_1$, as otherwise, for any fixed bids of the other players agent $i$ gets negative utility from bidding $b$ and positive utility from bidding $v_1$. The utility for player $j$ in this case is bounded from above by $v_j - (M - v_1) + \frac{M}{n-1} < 2v_1 - \frac{n-2}{n-1}M < 2v_1\big(1 - \frac{n(n-2)}{n-1} \big) < 0$, where the second inequality is since $M > 2 n v_1$ and the last inequality holds for any $n\geq 3$. That is, in every deviation by agent $j$ in which $j$ can win and get a positive utility \emph{from the auction} (by incentivizing agent $1$ to bid below $v_j$ with some probability 
in its bidding dynamic), the payment is too large, so the resulting utility is negative. So player $j$ prefers not to make payments and to get a utility of $\epsilon/(n-1)$.  
The case of $n=2$ can be shown with a slightly different argument. In this case, agent $1$ pays agent $2$ an amount of $M$ whenever agent $1$ bids $b \neq v_1$. 
First, under the no-regret property of agent $2$, agent $2$ can outbid $v_1$ only rarely (since any winning bid $b_2>v_1$ yields negative utility while bidding $0$ yields non-negative utility). 
Therefore, bidding $v_1$ gives agent $1$ a utility of at least $v_1-\epsilon$, and we claim that agent $1$ almost always wins the auction.   
To see this, notice that a bid by agent $1$ for which agent $2$ pays agent $1$ an amount less than $M$ is still dominated by bidding $v_1$ (which gives agent $1$ at least zero utility), and therefore such bids are not played by agent $1$ with payments less than $M$ by agent $2$. Suppose now that there is a bid $b$ for which agent $2$ pays agent $1$ an amount of $M$. When agent $1$ loses with the bid $b$, she gets a utility of zero. When agent $1$ wins, she has an expected payoff of $v_1 - \E[b_2|b > b_2]$. Denote by $q$ the probability that $b$ is a winning bid for agent $1$. We have that either $b$ is dominated by $v_1$, and thus not played by agent $1$, or 
$
v_1 - \epsilon \leq q(v_1 - \E[b_2|b > b_2]) \leq qv_1,
$ 
so we have 
$
q \geq 1 - \frac{\epsilon}{v_1}. 
$ 
And so agent $1$ wins at least $1-\epsilon$ fraction of the time.  
Now, consider agent $2$. By the argument above, agent $2$'s payoff from using bids that are above zero is at most 
$v_2 \cdot \epsilon < \epsilon$. Therefore, using such bids gives agent $2$ positive regret compared to bidding zero (when she bids zero, she gets a payment of $\epsilon$). Thus, agent $2$ almost always bids zero. It follows that player $2$ cannot make payments that increase her utility, so the policy of always paying zero is a best response for player $2$.  
\end{proof}

We are now ready to complete the proof. With the policy profile that we have, player $1$ gets a utility of $v_1 - \epsilon$. By taking a sufficiently small value of $\epsilon$,  player $1$ gets a utility per time step that is arbitrarily close to $v_1$. Since $v_1$ is the highest possible welfare in the game and player $1$ gets $\epsilon$ close to it for all $\epsilon>0$, this is an approximate best response for player $1$, and we have an $\epsilon$-Nash equilibrium. The auctioneer's revenue in this equilibrium is zero, other than vanishing profits from the initial learning phase of the agents.     
\end{proof}

\subsection{First-price auctions}
In contrast to the second-price auction, the first-price
auction is known not to be incentive compatible. Intuitively, the
direct dependence of the utility on the bid requires players to
constantly optimize their responses to the bids of the other players.
There has thus been significant interest in learning strategies,
outcomes, and dynamics in this setting 
(see Section \ref{sec:related-work} for references). Interestingly, despite the
significant difference from the truthful second-price auction, it is
well known (and not hard to see) that the Nash equilibrium of the
first-price auction yields the truthful outcome of the second-price
auction. 
It has been shown that 
mean-based  
(as defined in \cite{braverman2018selling}, see Appendix \ref{sec:appendix-definitions})
regret-minimization dynamics \cite{kolumbus2022auctions}, as well as best-response dynamics \cite{nisan2011best} can converge only to this outcome.  

In \cite{feldman2016correlated}, it was shown that the first-price auction also
has ``collusive'' coarse correlated equilibria\footnote{A generalization of the notion of correlated equilibrium, equivalent to all players having no external regret; also known as the Hannan set. See Appendix \ref{sec:appendix-definitions} for references.} (CCE) with lower revenue for the auctioneer than the second-price outcome, and higher utility for the players. The existence of such equilibria poses an opportunity for players to try and reach cooperative outcomes and increase their payoffs (at the expense of the auctioneer). 
However, to the best of our knowledge, no natural process has been described in nearly a decade since that result that attains any of these equilibria.

Moreover, in \cite{kolumbus2022auctions}, it is shown that even a scenario where players strategically manipulate the objective functions of their learning algorithms is not sufficient to reach any form of cooperation. 
In fact, convergence of standard no-regret agents to the second-price outcome induces a dominant strategy for users to provide the true objectives for their agents. 
%%% 
To understand how cooperative outcomes could perhaps still be reached by learners in the first-price auction, we observe that in the collusive equilibria
described in \cite{feldman2016correlated}, the low-value player wins in some 
fraction of the auctions and has positive utility, whereas under
simple no-regret dynamics (as those studied in~\cite{kolumbus2022auctions}), only a single player wins in the limit. This suggests that payments between the agents during their dynamics could enable the agents to share the welfare and reach cooperative outcomes. Yet, the existence of such equilibria does not by itself explain how they could arise under learning; the question is whether players have incentives to induce them.

The following result shows that for a large class of learning algorithms, players have an incentive to use payments in a first-price auction, and in equilibria of the payment game, the agents reach cooperative outcomes, with a significant reduction in revenue for the auctioneer.  

\begin{theorem}\label{thm:first-price-cooperative-outcomes}
(1) The single-item first-price auction with $n$ players and any mean-based no-regret bidding agents is not stable: players prefer to use payments. (2) In the two-player case, every equilibrium of the payment game is a strong Pareto improvement (i.e., both players are  better off) and has lower revenue for the auctioneer compared to the game without payments. 
\end{theorem}

\begin{remark}
The first part of Theorem \ref{thm:first-price-cooperative-outcomes}, showing that first-price auctions are not stable, applies to any number of players by a reduction to a game between the two top-valued players. In this reduction, all the lemmas extend with only minor changes in their proofs. For simplicity of presentation, we give below the proof outline for the two-player case. See Appendix \ref{sec:appendix-FP-auction-n-player-NE} for further details. 
\end{remark}

The proof proceeds via a sequence of lemmas. We sketch the argument and state the lemmas here; the full formal details appear in Appendix~\ref{sec:appnedix-Thm2-proof}.  
The first part of the proof is the construction of a unilateral payment policy for the high-valued player and the analysis of the resulting dynamics. 
Under this policy, the agent with the higher value pays $\eta>0$ to the lower-value agent whenever the latter bids zero. 
We derive the (unique) mixed Nash equilibrium of the game with this
payment policy, with $\eta$ as a parameter (Lemma \ref{thm:lemma-first-price-NE-with-eps-payment}). We then evaluate the
utilities of the players in that equilibrium and show that the optimal
value of $\eta$ for the high-valued player is half the value
of the second player (Lemma \ref{thm:lemma-first-price-NE-utilities}),
and that in this case, the players have higher utilities than in the
game without payments. 
We next prove that in the game between the agents, while the agents do not play a mixed equilibrium and their strategies are correlated, the marginal distributions and payoffs are nevertheless the same as in the mixed Nash equilibrium we derived (Lemma \ref{thm:lemma-first-price-mean-based-NE-CDFs}). Finally, using the fact that these utilities for the players are achieved through a unilateral payment by the high-valued player, we show that they provide a lower bound on the utilities in any equilibrium of the payment-policy game. 
We now state the lemmas formally.

\begin{lemma}\label{thm:lemma-first-price-NE-with-eps-payment}
Consider the single-item first-price auction with two players and player values $v_1 \geq v_2$, 
where player $1$'s payment policy is such that agent $1$ pays $\eta$ to 
agent $2$, where $0 < \eta < v_2$, whenever agent $2$ bids zero.
%%% 
The resulting game between the agents has a unique Nash
equilibrium where the cumulative density functions of the bids of
agents $1$ and $2$ are {\small$
F_1(x) = \frac{\eta}{v_2 - x} $}
and 
{\small$ G_2(x) = \frac{v_1 - v_2 + \eta}{v_1 - x}, $} respectively (so the bid density functions are supported on $[0, v_2 - \eta]$).
\end{lemma}

%%%%%%%%%%%%%%%%%%%%%%%%%%%%%%%
\begin{lemma}\label{thm:lemma-first-price-NE-utilities}
    If the agents 
    reach the Nash equilibrium
    outcomes, the optimal 
    value of $\eta$ for player $1$ is $v_2/2$. With this value 
    of $\eta$, the utilities for the players are {\small $v_1 - v_2 +
    \frac{v_2^2}{4v_1}$} for player $1$ and {\small$\frac{v_2}{2}$}
    for player $2$. The winning frequency of player $2$ is in the interval $[0,3/8]$, with the lower bound attained as $v_2/v_1 \rightarrow 0$ and the upper bound attained as $v_2/v_1 \rightarrow 1$.
\end{lemma}

The result of Lemma \ref{thm:lemma-first-price-NE-utilities} is for 
Nash equilibrium. While, in principle, 
no-regret agents may play an arbitrary CCE and not necessarily this outcome, the next lemma shows that the family of Nash-equilibrium distributions
parameterized by $\eta$ captures well the marginal distributions and
utilities of the broad family of mean-based \cite{braverman2018selling} no-regret agents in these games.      

\begin{lemma}\label{thm:lemma-first-price-mean-based-NE-CDFs}
    Consider the game where player $1$'s agent pays $\eta$ to player $2$'s agent whenever the latter agent bids zero. Fix a CCE with marginal
    distributions $F_1$ and $G_2$ to which the dynamics of mean-based agents converge with positive
    probability. 
    Then,  
    {\small$F_1(b) = F_1^{NE}(b) =
    \frac{\eta}{v_2 - b}$, $G_2(b) = G_2^{NE}(b) = \frac{v_1 - v_2 +
    \eta}{v_1 - b}$}, and the players' utilities are the same as in the 
    Nash equilibrium. 
\end{lemma}

\begin{figure}[t]
% \vspace{-12pt}
\centering
\begin{subfigure}{.49\linewidth}    \includegraphics[width=1.07\linewidth]{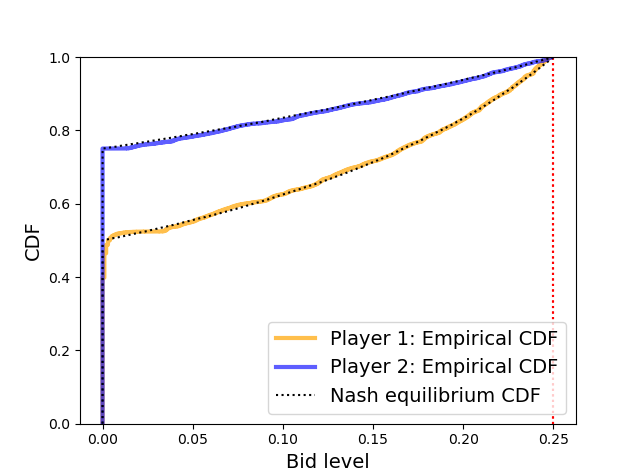}
    % \vspace{-9pt}
    \caption{{\small Cumulative distributions of bids}}  
    \label{fig:first_price_CDFs}
\end{subfigure}
\begin{subfigure}{.49\linewidth}
        \includegraphics[width=1.07\linewidth]{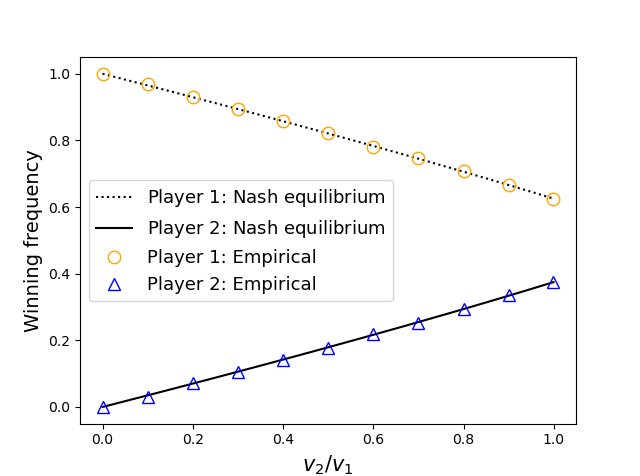}
        % \vspace{-9pt}
        \caption{{\small Winning frequencies}} \label{fig:first_price_winning_frequencies}
\end{subfigure}
\caption{Dynamics of bidding agents using Hedge algorithms with payments in first-price auctions.}
% \vspace{-1pt}
\label{fig:first_price_auction}
\end{figure}

Figure \ref{fig:first_price_auction} shows simulation results with
Hedge agents (a version of multiplicative weights \cite{slivkins2019introduction})
in a sequence of $100$,$000$ first-price auctions with payments between the agents. The left panel compares the Nash equilibrium theoretical
prediction from Lemma \ref{thm:lemma-first-price-mean-based-NE-CDFs}
with the empirical CDF from the agents' dynamics with $v_1 = 1,
v_2=0.5$, and agent $1$ pays agent $2$ $\eta=v_2/2$ when the latter bids zero, as in Lemma \ref{thm:lemma-first-price-NE-utilities}. The right panel compares the theoretical winning frequencies
of the players as specified in Lemma \ref{thm:lemma-first-price-NE-utilities} with their empirical winning
frequencies as a function of the ratio of values $v_2/v_1$ in multiple simulations. The long-term marginal bid distributions and winning frequencies are
clearly consistent with the theoretical predictions. We observe similar
results with other algorithm variants such as follow the perturbed leader
and linear multiplicative weights \cite{arora2012multiplicative}.
There is very little variation between simulation instances.  

Interestingly, it turns out that these dynamics of the agents with payments recover, as a special case, the known collusive CCE distribution of the auction without payments from \cite{feldman2016correlated}. 

\begin{observation}\label{thm:observation-min-revenue-CCE}
    For symmetric bidders with $v_1 = v_2 = 1$ and payment of $\eta =
    1/e$, the Nash equilibrium bid distribution of the 
    game with payment $\eta$  
    is the same as that of the minimum-revenue
    coarse correlated equilibrium of the standard first-price auction
    (without payments) given in
    \cite{feldman2016correlated}.  
\end{observation}
Notice, however, that in our game the players'
utilities are different from those of the CCE without payments. In particular, the payment $\eta$ that
results in this CCE distribution of play 
is not optimal for player $1$, so it is not an equilibrium of the
payment-policy game.    

\vspace{5pt}
The lemmas above then allow us to complete the proof of Theorem \ref{thm:first-price-cooperative-outcomes} by establishing a lower bound on the utility of the high-valued player in any Nash equilibrium of the payment game. The idea is to use the fact that the game analyzed in our proof is induced by a unilateral payment policy applied by the high-valued player. In any equilibrium, the player $1$ can consider the deviation in which she rejects (pays back) all the payments she received and makes this unilateral payment, which leads to a utility of {\small $v_1 - v_2 + v_2^2/4v_1$}. See Appendix \ref{sec:appnedix-Thm2-proof} for further details.
 
\bfpar{Welfare and revenue loss} To evaluate the potential loss in welfare in the game considered above (since sometimes a player with a lower value wins),
we can analyze the winning frequencies from Lemma
\ref{thm:lemma-first-price-NE-utilities}. The reduction in welfare is equal to the difference $v_1 - v_2$ multiplied by the
fraction of rounds in which the low-valued agent wins the auction. This welfare reduction is maximized, as can be verified by numerical calculation, when $v_2/v_1 \approx 0.50959...$, at which point the reduction is approximately $8.96\%$ of the full welfare. 
When considering the reduction in revenue, since the support of the bids is between zero and $v_2/2$, almost all the payments to the auctioneer are strictly less than $v_2/2$, so the revenue is reduced to less than half the revenue of the outcome of the game without payments. Thus, the players manage to improve their utilities at the expense of the auctioneer, with a relatively small loss in efficiency.

%% file: single_player.tex
\section{Manipulations by One Player}\label{sec:single-player-manipulations}
We now turn our attention to general games and focus on the instability of not using payments, by analyzing the effects of the simplest type 
of manipulations: payment policies where only a single agent pays. The question we aim to answer here is {\em when does a player have an incentive to use payments with  learning agents?} By making sufficiently high payments to all the other agents, clearly, an agent could induce any pure outcome in the game.
The main questions are thus not about the power of payment policies,
but about when are such policies also profitable, and what are the implications of such profitable payment policies for one player on the welfare of other players.

\bfpar{Notation}
Denote the welfare of an outcome $s$ by $w(s)$ and the
best-response utility of player $i$ to action profile $s$ by
{\small $u_i^{BR}(s) = \max_{s_i' \in S_i} u_i(s_i',s_{-i})$}. For the
following analysis, it is useful to define a notion of regret that measures the opportunity cost per step for player $i$ of playing the game where the strategy profile $x$ is played compared to the benchmark of her best action in hindsight in the (different) game where $y$ (rather than $x$) is played by the other players. To avoid confusion with standard  regret (which considers the benchmark of the best response in hindsight to $x$), we call this ``comparative regret.'' Intuitively, it measures the extent to which a player regrets not incentivizing others to play $y$ instead of $x$. That is, this quantity captures how much player $i$ would benefit, in hindsight, from having induced others to play differently.

\begin{definition}
    Let $x,y$ be two joint distributions on the players' joint action space $S$. The \emph{comparative regret} of player $i$ under the distribution of play $x$ compared to best responding to an alternative distribution of play $y$ is $R_i(x,y) = u_i^{BR}(y) - u_i(x)$.    
\end{definition}   

At time $T$, we have an expected distribution of play $x(T)$ with
welfare $w(x(T))$. We think of $x(T)$ as the distribution obtained by the agents without payments, and so $x(T)$ approaches the set of
CCEs of $\Gamma$ as $T \rightarrow \infty$. To
simplify notation, we omit the dependence on $T$ and discuss it
only when necessary. The optimal welfare is denoted $OPT = \max_{s\in S} w(s)$. 

\vspace{5pt}
The following theorem characterizes the cases in which there is a player in the game who can gain from incentivizing the other agents to reach the optimal-welfare outcome.  

\begin{theorem}\label{thm:single-player-optimal-welfare}
    Fix a finite game $\Gamma$ and a set $A$ of regret-minimizing agents. Let the optimal welfare outcome\footnote{To simplify the analysis, we consider the generic case of a unique optimal outcome. In  games with equalities in utility, this can be obtained by considering infinitesimal perturbations of the utilities.} of \ $\Gamma$ be $y^*$. If $y^*$ is a Nash equilibrium, then there is a player $i$ such that $i$ increases her own payoff and the 
    % payoffs 
    total social welfare  
    of all players by using a payment policy $p_i$ that induces $y^*$ as the unique long-term outcome of the agents' dynamics. If $y^*$ is not a Nash equilibrium, the same holds if there exists a player $i$ for which the welfare gap $OPT - w(x)$ is greater than $\sum_{j \neq i} R_j(x,y^*)$.  
\end{theorem}

\bfpar{Analysis}
To prove the theorem, we begin by establishing two lemmas that
quantify the costs and gains for a player of pushing the agents'
dynamics to an arbitrary pure outcome. 

To ensure the convergence of arbitrary no-regret agents to $y^*$, the payments must induce $y^*$ as the unique CCE,\footnote{With multiple CCEs, no-regret dynamics may fail to converge even in the time average \cite{KolumbusN22}.} which implies that it is also the unique Nash equilibrium. 
Since we consider $y^*$ to be an optimal welfare outcome, we may take $y^*$ to be a pure outcome of the game. One way to do so is by making $y^*$ an equilibrium in dominant strategies. Although weaker requirements would be sufficient for this purpose (see, for example, Proposition \ref{thm:iterated-dominant-strategies}), using dominant strategies simplifies the analysis. Moreover, it turns out that it leads to the same payments and the same utilities for the players in the long run as inducing a unique CCE: \cite{monderer-tenenholtz-2004} analyzed a mechanism-design scenario where a designer adds exogenous payments to the game, and they observed that the actual cost of inducing an outcome $s$ as an equilibrium in dominant strategies --- which they called the optimal \emph{$k$-implementation} of $s$, $k(s)$ --- is equal to the cost of inducing that outcome as a Nash equilibrium. Combining this with the fact that an equilibrium in dominant strategies is also the unique CCE, we have that, indeed, the exogenous payments needed for inducing a unique CCE or an equilibrium in dominant strategies are the same. The intuition for this result is that further payments that are intended to assure dominant strategies are in fact made only off the equilibrium path, and are thus not actually made when the game is played. A similar effect was also used in \cite{babaioff2022optimal} in the context of providing collateral contracts to mitigate strategic risk and incentivize investments.  

This intuition also carries over, to some extent, to our endogenous setting, where one of the players uses a payment policy to manipulate the agents' dynamics. But our case is different in several ways. First, and importantly, there are no external funds injected into the game. Instead, payments that a player allows her agent to make to other agents are taken from her own payoffs. Second, as we show below, the cost of inducing an outcome $s$ as an equilibrium is lower than $k(s)$, since the player who wishes to induce $s$ as the outcome of the agents' dynamics does not make payments to her own regret-minimizing agent. Third, since in our case the learning agents do not reason about payment policies but rather learn to respond over time, some agents will play dominated actions for some time, leading to additional short-term costs. This is formalized in the following lemma. 

\begin{lemma}\label{thm:lemma-single-player-cost}
    Fix a player $i$ and pure outcome $y$. Player $i$ can incentivize $y$ as the unique long-term outcome of the dynamic with an expected cost per time step of  
    $ \ cost_i(y) = k(y) - R_i(y,y) + o(1)$.
\end{lemma}

%%%%%%%%%%%%%%%%
\begin{proof} (Lemma \ref{thm:lemma-single-player-cost}). 
We start constructing a payment policy for the agent of player $i$ by looking at the optimal $k$-implementation payments for inducing $y$ as a dominant strategy equilibrium; we will then adjust this payment profile to fit our setting in which there are no external payments and to minimize player $i$'s cost. A key point is that the payment amount $k(y)$ includes payments to agent $i$ that an external designer would have made for inducing $y_i$ as the dominant strategy for agent $i$. These payments are not feasible and not needed in our case. Instead, we change the payment policy of agent $i$ in the following way: for any action $s_i \neq y_i$, agent $i$ makes a payment of $\max_{s\in S}u_i(s) + 1$, and distributes this payment equally to the other agents. These payments make all the other actions of agent $i$ strictly dominated by the action $y_i$. Thus, to calculate the cost for player $i$, we need to subtract from the $k$-implementation payments the payments made to player $i$, $R_i(y, y)$, and add a learning-phase error term $\tilde{c}$ that results from the additional payments that agent $i$ makes when playing its dominated strategies, as well as from other times during the learning dynamics when the other agents play their dominated strategies. 
\begin{claim}
    The expected learning phase cost per step $\tilde{c}$ vanishes in the long term, i.e., $\tilde{c} = o(1)$.
\end{claim}
\begin{proof}
    $\tilde{c}$ is the expected total payment per time step that player $i$ makes when  agents play strictly dominated strategies. Assume for contradiction that $\lim_{T\rightarrow \infty} \tilde{c} > 0$, or that the limit does not exist. Both cases imply that the frequency at which dominated strategies are played in the agents' dynamic is not vanishing, which contradicts the regret-minimization property of the agents.  
\end{proof}
It follows that $ cost_i(y) = k(y) - R_i(y,y) + o(1)$ as $T \rightarrow \infty$, as stated in the lemma.
\end{proof}
%%%%%%%%%%%%%%%%

The next lemma specifies the condition under which, given a distribution $x$ to which the no-regret agents converge in expectation in $T$ rounds, there is a player who can use a payment policy to increase her payoff by pushing the dynamic to a different outcome.  
\begin{lemma}\label{thm:lemma-single-player-conditions}
    Fix a finite game $\Gamma$ and a set of regret-minimizing agents for the players. If there exists a player $i$ and outcome $y$ such that  
    $w(y) - w(x) > \sum_{j \neq i} R_j(x,y)$, then $i$  can increase her payoff by making payments that induce $y$ as the unique outcome of the agents' dynamics. 
\end{lemma}
\noindent
Intuitively, the increase in welfare must be large enough to compensate the other players for their regret for playing $x$ compared to their best responses to the alternative outcome $y$. Before proving the lemma, we have the following corollary. 

\begin{corollary}\label{cor:single-player-deviation-to-better-NE}
   If the dynamics without payments approach a welfare $w(x)$ lower than the welfare of some Nash equilibrium $y$, then, since in equilibrium $R_i(y,y) = 0$, there is a player who can increase her payoff by using a payment policy that induces $y$ as the long-term outcome of the dynamics. 
\end{corollary}

%%%%%%%%%%%%%%%%%%
\begin{proof}(Lemma \ref{thm:lemma-single-player-conditions}).
    We need to specify the requirement that player $i$ will increase her own utility after making the payments that induce $y$ as the unique outcome of the dynamics. 
    That is, $cost_i(y) < u_i(y) - u_i(x)$.
    By Lemma \ref{thm:lemma-single-player-cost} we have that for large enough $T$, player $i$ prefers making such payments if 
    $u_i(y) - k(y) + R_i(y,y) > u_i(x)$. 
    Rewriting the left terms we have 
    $$
    u_i^{BR}(y) - k(y) > u_i(x) = w(x) - \sum_{j \neq i} u_j(x).
    $$
    Using the fact that by definition of $k(y)$ \cite{monderer-tenenholtz-2004}, it can be expressed as $$k(y)=\big[ \sum_{j} u_j^{BR}(y)\big] - w(y),$$ 
    we can rewrite the left-hand side as
    $
    w(y) - \sum_{j \neq i} u_j^{BR}(y).
    $
    Rearranging the terms, we have 
    $$
    w(y) - w(x) > \sum_{j \neq i} u_j^{BR}(y) - \sum_{j \neq i} u_j(x) = \sum_{j \neq i} R_j(x,y). 
    $$  
\end{proof}%
Note that all the steps in the proof of Lemma~\ref{thm:lemma-single-player-conditions} hold in either direction. Thus, the lemma, and consequently the result of Theorem~\ref{thm:single-player-optimal-welfare} for the case where the optimal welfare outcome is not an equilibrium, also hold in the other direction. That is, if player $i$ can increase her payoff by making payments that induce $y$ as the unique outcome, then $
w(y) - w(x) > \sum_{j \neq i} R_j(x, y).
$

%%%%%%%%%%%%%%%%%%
\vspace{5pt}
\noindent
The lemmas above provide the basis for proving Theorem \ref{thm:single-player-optimal-welfare}.

\begin{proof}(Theorem \ref{thm:single-player-optimal-welfare}).
If the dynamics of the players'  regret-minimizing agents converge to the optimal-welfare outcome, $lim_{T \rightarrow \infty} x = y^*$, then there is no room for improvement and the theorem holds trivially (although weakly) with payments of zero. In the following, we assume for simplicity that $w(x)$ is bounded away\footnote{There is also the scenario where $x = y^*$ for some values of $T$ but not in the limit. For such values of $T$ the theorem holds trivially as well, and so the assumption is without loss of generality.} from $w(y^*)$, that is, there exists $\epsilon > 0$ and $T_0$ such that for all $T>T_0$, we have that $w(x) < w(y^*)(1 - \epsilon)$. 

Suppose that $y^*$ is a Nash equilibrium. By  definition, if some profile $s$ is a Nash equilibrium, then $R_i(s,s) = 0$ for all $i$. Thus, by Lemma \ref{thm:lemma-single-player-cost}, any player can construct a payment policy that induces $y^*$ as the unique outcome of the agents' dynamic with a cost-per-time-step of $o(1)$. 

Therefore, to have a player $i$ who increases her payoff by inducing $y^*$ as the outcome of the dynamics, we only need to have $u_i(y^*) > u_i(x)$. Assume for the purpose of contradiction that for all $i$, $u_i(x) \geq u_i(y^*)$. Summing over the players, we have $\sum_i u_i(x) \geq \sum_i u_i(y^*)$; that is, $w(x) \geq w(y^*)$, a contradiction to the optimality of $y^*$. Thus, we have that for $T$ large enough such that the learning-phase costs are small, there exists a player $i$ who can increase her payoff by inducing the optimal-welfare outcome as the outcome of the agents' dynamics. 
If $y^*$ is a non-equilibrium outcome, the conditions for $i$ increasing her payoff are given by Lemma \ref{thm:lemma-single-player-conditions},  setting $y = y^*$.
Notice that when the player $i$ adopts this strategy to selfishly improve her own payoff, the resulting dynamics converge to a high-welfare outcome that is on the Pareto frontier of the original game.
% Notice also that since under this payment policy, the outcome $y^*$ dominates $x$ for all the agents, the utilities are a strict Pareto improvement compared to the game without payments.   
\end{proof}

From the analysis above, we observe that in many cases, players can gain from using payment policies to divert the agents' dynamics to more favorable outcomes for them.
The following theorem shows that this is true for a broad class of games.   
For games with a price of stability (PoS) of 1, this is already implied by Theorem
\ref{thm:single-player-optimal-welfare}, but the following result is
more general.  

\begin{theorem}\label{thm:single-player-stability}
    Any finite game in which there exists a CCE with lower welfare than the best 
    pure Nash equilibrium is not stable. This holds, in particular, if $PoA\neq PoS$ in pure strategies. 
\end{theorem}
\begin{remark}
    The converse is not true in general: there are games with $PoA=PoS$ that are not stable, so players prefer to use payments. One such example is the prisoner's dilemma game presented in the introduction  
    in Figure 
    \ref{fig:PD-game-a}.
\end{remark}
\begin{proof}
The result follows from combining the proof of Lemma \ref{thm:lemma-single-player-conditions} with general properties of regret-minimization dynamics.
Let $\Gamma$ be a game in which the Nash equilibrium with the highest welfare, denoted $\Tilde{y}$, yields a welfare of $\Tilde{w}$ (where the welfare in equilibrium $\Tilde{w}$ is not necessarily equal to the optimal welfare $OPT$), and there exists a CCE $x$ with welfare $w(x) < \Tilde{w}$.

It is known that for any CCE of a game, there are regret-minimizing agents that converge to that equilibrium \cite{KolumbusN22,monnot2017limits}. In particular, there are agents that converge arbitrarily close to $x$.\footnote{The proofs of this convergence result use cycles of pure actions to approximate the time average of joint distributions, and the arbitrarily close approximation is due to the density of the rational numbers. We disregard this detail in our analysis.} By Lemma \ref{thm:lemma-single-player-conditions} (in particular, using Corollary \ref{cor:single-player-deviation-to-better-NE}), there exists a player who prefers using a payment policy that assures convergence to the highest welfare equilibrium $\Tilde{y}$ compared to not making payments and reaching outcome $x$. Thus, we have a utility-improving deviation over the no-payments action profile and hence the game is not stable. 
\end{proof}

%% file: two_players.tex
\section{Two-Player Games}\label{sec:two-player-games}

In the prisoner's dilemma example presented in Section \ref{sec:intro}, we saw that the players had  
incentives to use payments between their agents, and in equilibria of the payment game, both players had 
higher gains than in the game without payments.
We now show that these phenomena apply quite generally to bilateral interactions where the players use learning agents. 

\begin{theorem}\label{thm:two-players-Pareto-improvement}
    Every equilibrium of a two-player payment-policy game is a  weak Pareto improvement over the no-payments outcome of this game. 
\end{theorem}

\begin{proof} 
The intuition for this result is that in order to have any positive payment in equilibrium, some player must benefit from making the payment, while the other benefits from receiving it. If some player is worse off with the payments, then, since there are no considerations about other players, this player always has a deviation that effectively cancels the payments. With more than two players, there may be players who cannot profitably counteract such payments made among others that do not benefit them.
Formally, consider a game $\Gamma$, a set $A$ of regret-minimizing learning agents for the players, and a time horizon $T$. 

Suppose, for the sake of contradiction, that there exists an equilibrium $p^*$ of the payment-policy game associated with ${\Gamma, A, T}$ where player $i$ has a lower utility than her expected utility over the $T$ rounds in the game without payments. Let $h$ denote the utility for player $i$ in the game without payments, and let $l$ denote the utility for player $i$ in the equilibrium $p^*$. By  our assumption, we have $h > l$. Consider the following deviation from the payment policy $p^*$: the player $i$ matches her payments to zero out all payments made by the other player and, additionally, reduces all her other payments to zero. As a result, the two agents now observe in their dynamics the utilities of the original game $\Gamma$ without payments. Therefore, player $i$ has the utility of $h$. Since player $i$ has increased her utility from $l$ to $h$, we have a contradiction to $p^*$ being an equilibrium. Hence, the utilities of both players can only increase relative to the game without payments.    
\end{proof}

The next result addresses the question of incentives to use payments in a subclass of games in which no-regret dynamics converge to a pure outcome. We see that despite the simplicity of learning in these games, and even in the presence of dominant strategies, players may still have incentives to use payment policies. 

One intuition for this is that in two-player interactions, players can use payments to effectively assume the role of a Stackelberg leader in the game. 
In particular, in some games the equilibrium outcome to which learning dynamics converge may yield a lower payoff for some player than the Stackelberg outcome for that player, creating an incentive to use payments to induce such higher-payoff outcomes. This can arise even in games that are extremely simple in the no-payments setting, such as dominance-solvable games and games with iterated dominant strategies (see Appendix~\ref{sec:appendix-PD-game}). 
\begin{theorem}\label{thm:two-players-DS-games}
    Any finite two-player game with a unique coarse correlated equilibrium in pure strategies and at least two Stackelberg equilibrium outcomes with different payoffs (depending on which player is the Stackelberg leader) is not stable for any regret-minimizing agents for the players. 
\end{theorem}

\begin{proof} 
Consider as our starting point the action profile in which the players use zero payments. 
In the game without payments, any regret-minimization dynamic will converge to the unique coarse correlated equilibrium (which is also the unique Nash equilibrium). To show that the game is not stable for any set of regret-minimizing agents for the players, we need a utility-improving deviation from the zero-payment profile for one of the two players. 

We use the fact that when the opponent does not make payments, a player has the ability to alter her policy so as to induce the agent to playing a pure action in the long term; this can be obtained by making large payments whenever the agent plays different actions. Since for at least one of the players there is a Stackelberg-equilibrium outcome that is different from the (unique) Nash-equilibrium outcome, there exists a player who obtains a higher payoff as a Stackelberg leader in the game, so this player prefers to use such a policy. Since the other agent is also minimizing regret, this agent will eventually learn to best reply to the fixed action of the first player's agent, and the dynamics with these payments will converge to the Stackelberg outcome of the game. 
Thus, we have a utility-improving deviation from the zero-payment policy, so the game is not stable. 
Note that this is a different result from that of Theorem \ref{thm:single-player-stability}, since the Stackelberg outcome here is not a coarse correlated equilibrium of the game $\Gamma$. 
\end{proof}

%% file: conclusion.tex
\section{Conclusion}\label{sec:conclusion}
Autonomous learning agents are widely and increasingly used in many online economic interactions, such as auctions and other markets. Rapid advancement in large language models and their integration into a wide array of applications further reduce barriers to communication and interaction among autonomous agents. It is not difficult to envision a future digital landscape populated by a plethora of autonomous agents that will be endowed with access to financial assets or fiat currencies, trading among themselves to serve the interests of their users. Motivated by this transition to automated systems with sophisticated AI agents, we studied the scenario where users let their automated agents offer monetary transfers to other agents during their dynamics. 

Our results show that strategic users quite generally have incentives to use payments with their learning agents, which can have very significant implications for the overall outcomes. 
While the focus in this work is on the general properties of payments between learners and their analysis in the context of auctions, studying payments among learning agents may also be of interest in a broad range of other markets and strategic interactions, including  Fisher market models, asset markets, or principal-agent games with multiple learning agents. 

These incentives for utilizing payments between agents mark a potential risk for market environments. While payments between competing firms in classical markets are constrained by antitrust and competition law, learning agents with financial autonomy might engage in transfers whenever doing so improves their objective. This raises new challenges for defining and preventing undesirable algorithmic behavior, assigning accountability for agents’ actions, and regulating the financial autonomy of learning agents in market settings --- challenges that are inherently tied to the strategic and learning dynamics induced by such transfers. A concrete example in which  payments by autonomous algorithmic agents are already widespread and largely unregulated is Maximal Extractable Value (MEV) in blockchain-based markets (see, e.g., the  ESMA and FCA reviews on that matter~\cite{esma_mev_2025,fca_mev_2024}). In this setting, trading bots generate action-contingent payments to incentivize execution priority for specific currency transactions, enabling a range of arbitrage and price-manipulation strategies, some of which would be considered illegal in traditional markets.

More broadly, our findings highlight a fundamental challenge for mechanism design in the AI era: understanding the interplay between the joint learning dynamics of learning-based agents, the ensuing strategic incentives of their users, and crucially, the desired mandates and regulatory principles that should govern their interaction in automated markets.

%% file: appendix_PD_example.tex
\section{Prisoner's Dilemma Games}\label{sec:appendix-PD-game}

We begin by stating a simple result for a class of games we term games with \emph{``iterated dominant strategies.''} This class generalizes games with dominant-strategy equilibria and forms a subclass of dominance-solvable games (see, e.g., \cite{jafari2001no}). In these games, there exists a sequence of iterated elimination of strictly dominated strategies that follows an order over the players, where in each step the player in turn eliminates all but one of their strategies. 
In particular, this class includes the prisoner's dilemma game. 
\begin{proposition}\label{thm:iterated-dominant-strategies}
    In a game with iterated dominant strategies, the empirical distribution of play of any type of regret-minimization dynamic converges to the unique Nash equilibrium.  
\end{proposition}
\begin{proof}
     Consider the following induction argument. The game has at least one order over the players of iterated elimination of strictly dominated strategies. Index the players according to one such order, such that player $1$ has a dominant strategy, player $2$ has a dominant strategy in the sub-game where player $1$ plays her dominant strategy, and so on. Additionally, index the corresponding dominant strategies of the players using the same order, so that $a_1$ is the action that is not eliminated by player $1$, $a_2$ is not eliminated by player $2$, and so on. Clearly, in any CCE, player $1$ plays her dominant strategy $a_1$; otherwise she has positive regret. Our induction step is simple: Fix any CCE. Given that players $1,...,k-1$ play actions $a_1,...,a_{k-1}$, playing $a_k$ is a strict best response for player $k$. By induction, this is true up to the last player $n$, so we have that our arbitrary CCE is the pure Nash equilibrium profile $(a_1,...,a_n)$. Therefore, regret-minimization dynamics converge to it. 
\end{proof}

\vspace{5pt}
\noindent
\emph{Equilibria of the game from Figure \ref{fig:PD-game-a}:} To approach the outcome described in the introduction as an equilibrium, we need to adjust the payment policy of player $1$ as follows. Player $1$'s agent  pays player $2$’s agent the maximum payment whenever player $1$’s agent cooperates, and pays $\nicefrac{1}{3} + \epsilon$ when the outcome $(D,C)$ is obtained. Effectively, player $1$ blocks player $2$ from incentivizing player $1$’s agent to cooperate, so player $2$’s best response is not to pay anything and get $\nicefrac{1}{3} + \epsilon$. This is an $\epsilon$-equilibrium for all $\epsilon > 0$,  since player $1$ incentivizes outcome $(D,C)$ in the limit $T\rightarrow \infty$, and she cannot improve her long-term payoff by more than $\epsilon$. The other equilibrium, as mentioned, is the mirror image of the same equilibrium (replacing the payment policies of players $1$ and $2$).

\begin{figure}[t!]
\centering
\begin{subfigure}{.49\linewidth}
    \centering
    \begin{NiceTabular}{cccc}[cell-space-limits=3pt]
         &     & \Block{1-2}{{\small Player $2$}} \\
         &     & $C$     & $D$ \\
    \Block{2-1}{{\small Player $1$}} 
         & $C$ & \Block[hvlines]{2-2}{}
               $ \ y, y \ $ & $ \ 0, x \ $ \\
         & $D$ & $ \ x, 0 \ $ & $ \ 1, 1 \ $ 
    \end{NiceTabular}
    \caption{{\small The symmetric game: $x > y > 1$.}}
    \label{fig:PD-game-b}
\end{subfigure}
\begin{subfigure}{.49\linewidth}
    \centering
    \begin{NiceTabular}{cccc}[cell-space-limits=3pt]
         &     & \Block{1-2}{{\small Player $2$}} \\
         &     & $C$     & $D$ \\
    \Block{2-1}{{\small Player $1$}} 
         & $C$ & \Block[hvlines]{2-2}{}
               $ \ y_1, y_2 \ $ & $ \ 0, x_2 \ $ \\
         & $D$ & $ \ x_1, 0 \ $ & $ \ 1, 1 \ $ 
    \end{NiceTabular}

\caption{{\small The asymmetric game: $x_i > y_i > 1$.}}
\label{fig:PD-game-asymmetric}
\end{subfigure}
\caption{Parametrized Prisoner's Dilemma games.}
\end{figure}

\vspace{8pt}
\noindent
\emph{Price of Stability and Price of Anarchy:} Consider the symmetric game in Figure \ref{fig:PD-game-b}, where in the prisoner's dilemma we have $x > y > 1$. The equilibria of the game depend of the welfare gap between the game outcomes. There  are two cases. If $x \leq 2$, then the price of anarchy is trivially bounded by $2$, as this is the maximum utility gap in the game. If $x > 2$, the payment game has two equilibria in which one agent pays the maximum amount when cooperating and pays $1 + \epsilon$ when the other agent cooperates, and the other agent does not make any payments. Each of these equilibria yields a welfare of $x$. Since $x > y$, we have $\nicefrac{2y}{x}  < 2$, so the price of anarchy is at most $2$.

In the asymmetric game, we assume without loss of generality that $x_1 \geq x_2$. We have three cases to analyze. If $x_1,x_2 \leq 2$, then the players do not have profitable payment policies, so, as before, $PoA=PoS$ and both are trivially bounded by $2$. 
If $x_1 > 2$ and $x_2 \leq 2$, then only player $1$ has a profitable payment policy and the payment game has a single equilibrium with the outcome $(D,C)$, which yields a welfare of $x_1$ Since $x_1 > y_1, y_2$, we have $\frac{y_1+ y_2}{x_1} < 2$, $PoA = PoS$, and both bounded by $2$. Finally, if $x_1, x_2 > 2$, as in the symmetric game, we have two equilibria for the payment game with outcomes $(D,C)$ or $(C,D)$. The $PoS$ in this case is at most $2$ by the same argument as in the previous case. However, the ratio $x_1/x_2$ between the two equilibria can be unbounded, so the $PoA$ can be unbounded.

%% file: appendix_FP_proof.tex
\section{Deferred Proofs for Theorem \ref{thm:first-price-cooperative-outcomes}}\label{sec:appnedix-Thm2-proof}
\begin{proof} (Lemma \ref{thm:lemma-first-price-NE-with-eps-payment}). 
We start by assuming that in the agents' game, parameterized by $\eta$, there is a mixed Nash equilibrium in which each player's strategy has a possible atom at $0$ and is otherwise absolutely continuous on $(0, v_2-\eta]$ with a strictly positive density. Equivalently, the cumulative distribution function is continuous on $(0, v_2-\eta]$ and may have a jump at $0$.  
The validity of this ansatz will be verified shortly. 

Consider player $1$'s utility. In equilibrium, she is indifferent between bidding zero and bidding $v_2 - \eta$, giving us the size of the point mass for player $2$ at zero:
$$
v_1 - (v_2 - \eta) - G_2(0)\eta = G_2(0)(v_1 - \eta)
\quad \Rightarrow \quad
G_2(0) = 1 - \frac{v_2 - \eta}{v_1}\jhedit{.}
$$
Player $1$'s utility when bidding $x$ is $u_1(x,y) = G_2(x)(v_1-x) - \chi(y)\eta$, where $\chi(y)$ is an indicator which equals $1$ if player $2$ bids zero (i.e., if $y=0$) and equals zero otherwise. Importantly, it is independent of $x$.
Player $1$ is indifferent between all her bids, so $\frac{\partial u_1(x,y)}{\partial x} = 0$, which gives  
$$
G'_2(x)(v_1 - x) - G_2(x) = 0. 
$$
We have an ordinary differential equation with a unique solution of the form 
$
G(x) = \frac{c}{v_1 - x}.
$
Given our initial assumption, we must have continuity at zero, and thus we have 
$
c = v_2 G_2(0) = v_1 - v_2 + \eta
$.
Verifying the consistency of the support (the CDF must be equal to $1$ at the top of the support):
$$
G_2(v_2 - \eta) = \frac{v_1 - v_2 + \eta}{v_1 - (v_2 - \eta)} = 1. 
$$

Since player $2$ is indifferent between bidding $x>0$ and bidding zero for a utility of $\eta$, we have 
$$\eta = F_1(x)(v_2 - x)
\quad \Rightarrow \quad
F_1(x) = \frac{\eta}{v_2 - x}.
$$
Thus, we have a mixed Nash equilibrium with CDFs
$$
F_1(x) = \frac{\eta}{v_2 - x},
\quad  
G_2(x) = \frac{v_1 - v_2 + \eta}{v_1 - x},
$$ as stated in the lemma.
\end{proof}

\begin{proof} (Lemma \ref{thm:lemma-first-price-NE-utilities}).
We start by computing the expected payoff of player $1$, which is given by 
    $$
    \E[u_1] = \int_0^{v_2-\eta} 
    \Big( 
    G_2(x)(v_1 - x) - \eta G_2(0)
    \Big)
    f(x) dx,
    $$
    where the PDF is $f(x) = \frac{\eta}{(v_2 - x)^2}$. From the distributions we calculated, we get that
    $$
    \E[u_1] = 
    v_1 - v_2 + \eta 
    \Big(\frac{v_2 - \eta}{v_1}\Big).
    $$
    A first-order condition with respect to $\eta$ shows that the expectation is maximized when $\eta = v_2/2$, in which case its value is $v_1 - v_2 + v_2^2/4v_1$, as claimed.  With this choice of $\eta$, player 2's utility is $v_2/2$.

    The frequency at which the low agent (agent $2$) wins is given by:
    $$
    Pr[b_2 > b_1] = 
    \int_0^{v_2-\eta}
    F(x)
    g(x) dx = 
    \frac{v_2(2v_1 - v_2)}{4(v_1 - v_2)^2}
    \Bigg(
    \ln\Big(\frac{2v_1 - v_2}{v_1}\Big) - 
    \frac{v_2(v_1 - v_2)}{v_1(2v_1 - v_2)}
     \Bigg).
    $$
    If $v_2$ is small compared to $v_1$ (in the limit $v_2/v_1 \rightarrow 0$), agent $2$ never wins. If $v_2$ is close to $v_1$ (in the limit $v_2/v_1 \rightarrow 1$), agent $2$ will win $3/8$ of the time. 
\end{proof}

\begin{proof} (Lemma \ref{thm:lemma-first-price-mean-based-NE-CDFs}) We begin by showing that the Nash equilibrium cumulative distributions bound the marginal distributions in our dynamics from below (i.e., the latter stochastically dominate the Nash equilibrium).  
    \begin{claim}
        The support of $F_1$ and $G_2$ is at most $v_2 - \eta$, and for all $b$ in the support, we have that $F_1(b) \geq  F_1^{NE}$, and $G_2(b) \geq  G_2^{NE}(b)$.
    \end{claim}
    \begin{proof}
        For agent $2$, all bids above $v_2-\eta$ are strictly dominated by bidding zero, and so are not in the support. 
        For agent $1$, given the support of agent $2$, bids above $v_2 - \eta$ give strictly less utility than bidding $v_2 - \eta$, so a regret-minimizing agent will not play them.

        For mean-based agents, we have that for all $b$ in the support,   
        $u_2(b) = F_1(b)(v_2 - b) \geq u_2(0) = \eta$ and 
        $u_1(b) = G_2(b)(v_1 - b) - \eta G_2(0) \geq 
        v_1 - v_2 + \eta - \eta G_2(0)$. 
        Hence we have $F_1(b) \geq \frac{\eta}{v_2 - b} = F_1^{NE}$, and $G_2(b) \geq \frac{v_1 - v_2 + \eta}{v_1 - b} = G_2^{NE}(b)$.
    \end{proof}

    Continuing the proof of Lemma \ref{thm:lemma-first-price-mean-based-NE-CDFs}, suppose, for contradiction, $F_1(b') > F_1^{NE}(b')$ for some $b'< v_2-\eta$. We have that $F_1(b') > F_1^{NE}(b')$ implies $u_2(b') = F_1(b')(v_2-b') > F_1^{NE}(b')(v_2-b') = \eta = u_2(0)$. Therefore, bidding zero is not within the support of player $2$.
        
    However, if player $2$ never bids zero, then there are no payments in our CCE. Since the agents minimize regret, it must be a CCE of the standard first-price auction.

    In the standard first-price auction (without payments), mean-based agents can converge only to CCEs that include bids above $v_2 - \eta$ in their support; this follows from  \cite{kolumbus2022auctions}. This is a contradiction to our assumption. Therefore, we have $F_1 = F_1^{NE}$. 
    
    Next, suppose for contradiction that $G_2(b') > G_2^{NE}(b')$ for some $b' < v_2-\eta$. 
    This implies $u_1(b') = G_2(b')(v_1-b') - \eta G_2(0) > G_2^{NE}(b')(v_1-b') - \eta G_2(0) =  v_1 - v_2 + \eta -  \eta G_2(0)  = u_1(v_2 - \eta)$. That is, bidding $b'$ gives player 1 a higher utility than bidding $v_2 - \eta$, so the latter is not within the support of player 1, which is a contradiction, since the marginal distribution of player $1$ is $F_1$, which is supported on $[0, v_2-\eta]$.
    Thus, overall, we have $F_1 = F_1^{NE}$ and $G_2 = G_2^{NE}$.  
\end{proof}

\begin{proof} (Theorem \ref{thm:first-price-cooperative-outcomes}.)   
To prove the theorem, we use the fact that the game analyzed in the above  lemmas is induced by a unilateral payment policy by the high-value player.
In any equilibrium, player $1$ can consider the deviation to a policy in which she effectively rejects (pays back) all the payments she received and makes this unilateral payment, leading to a utility of $v_1 - v_2 + v_2^2/4v_1$. Therefore, in any equilibrium, player $1$ gets at least this utility, which is higher than the utility of $v_1 - v_2$ that she gets in the game without payments. As for the utility of the low-value player, since the high-value player's utility is above the second price, player $1$ must use positive payments, otherwise the dynamics of mean-based agents could approach only the second-price outcome (as shown in \cite{kolumbus2022auctions}). Therefore, in the equilibrium of the payment-policy game, the low-value player has positive utility, an improvement over the game without payments. 
Having established this, the fact that the revenue decreases in any equilibrium is now straightforward.  To see this explicitly, suppose for contradiction that the revenue $R$ is at least the revenue $v_2$ of the game without payments. The total welfare (of the players and the auctioneer) is $ w = u_1 + R + u_2 \leq v_1$. We then get  $v_1 + \eta\Big(\frac{v_2 - \eta}{v_1}\Big) + u_2 \leq v_1$, which is a contradiction, since player $2$ cannot get negative utility in equilibrium. 

By Lemmas~\ref{thm:lemma-first-price-NE-utilities}, \ref{thm:lemma-first-price-mean-based-NE-CDFs}, player $1$ has a utility-improving deviation from the zero-payments profile for any mean-based regret-minimizing agents for the players, so the game is not stable for such agents; also, both players improve their utilities, and thus the revenue for the auctioneer decreases.
\end{proof}

%% file: appendix_FP_n_player_NE.tex
\section{Extension of Lemmas for Theorem \ref{thm:first-price-cooperative-outcomes}}\label{sec:appendix-FP-auction-n-player-NE}
To extend the lemmas of Theorem \ref{thm:first-price-cooperative-outcomes} to the case of $n$ players, consider the following analysis. Suppose we have $n$ players with values $v_1>v_2>v_3 > \dots > v_n$, each using a regret-minimizing agent. If agent $2$ bids $v_3$, agent $1$ pays her $\eta > 0$.  
Essentially, this reduces to the same proof after shifting the support of the bids from $0$ to $v_3$. 
Lemma \ref{thm:lemma-first-price-mean-based-NE-CDFs} and \ref{thm:lemma-first-price-NE-utilities} have the following simple adjustments.

As in the proof of Lemma  \ref{thm:lemma-first-price-NE-with-eps-payment}, we begin by assuming that in the game between these agents, parameterized by $\eta$, there is a mixed Nash equilibrium in which each player's strategy has a possible atom at $v_3$ and is otherwise absolutely continuous on $(v_3, v_2-\eta]$ with a strictly positive density. Equivalently, the cumulative distribution function is continuous on $(v_3, v_2-\eta]$ and may have a jump at $v_3$. 

We start by looking at player $1$'s utility. In equilibrium, she is indifferent between bidding $v_3$ and bidding $v_2 - \eta$, giving us the (point-mass) probability that player $2$ bids $v_3$:
$$
v_1 - (v_2 - \eta) - G_2(v_3)\eta = G_2(v_3)(v_1 - v_3 - \eta)
\quad \Rightarrow \quad
G_2(v_3) = \frac{v_1 - v_2 + \eta}{v_1 - v_3}.
$$
Player $1$'s utility when bidding $x$ is $u_1(x,y) = G_2(x)(v_1-x) - \chi(y)\eta$, where $\chi(y)$ is an indicator which equals $1$ if player $2$ bids $v_3$ and equals zero otherwise. As before, note that it is independent of $x$.
Player $1$ is indifferent between all her bids, so $\frac{\partial u_1(x,y)}{\partial x} = 0$, which gives  
$
G'_2(x)(v_1 - x) - G_2(x) = 0
$. 
We have an ordinary differential equation with a unique solution of the form
$
G(x) = \frac{c}{v_1 - x}
$. 
From continuity above $v_3$, we get 
$
c = (v_1 - v_3) G_2(v_3) = v_1 - v_2 + \eta
$. 
Checking the consistency of the support (the CDF must be equal to $1$ at the top of the support):
$$
G_2(v_2 - \eta) = \frac{v_1 - v_2 + \eta}{v_1 - (v_2 - \eta)} = 1. 
$$

Player $2$ is indifferent between bidding $x>v_3$ and bidding $v_3$ and getting a utility of $\eta$, so we have $
\eta = F_1(x)(v_2 - x)
\quad \Rightarrow \quad
F_1(x) = \frac{\eta}{v_2 - x}
$. 
Thus, we have the CDFs
$$
F_1(x) = \frac{\eta}{v_2 - x},
\quad \quad 
G_2(x) = \frac{v_1 - v_2 + \eta}{v_1 - x}.
$$

For players $i > 2$, since there is a positive probability that both players $1$ and $2$ bid $v_3$ at the same time, any bid above $v_3$ wins the auction with positive probability and gives negative expected utility, whereas any bid $b\leq v_3$ gives a utility of zero. Therefore, any bid distribution for players $i>2$ that has zero weight on bids above $v_3$ forms a Nash equilibrium, together with the distributions $F_1,G_2$ for the first two players.  
The expected payoff of player $1$ is given by 
$
\E[u_1] = \int_{v_3}^{v_2-\eta} 
\Big( 
G_2(x)(v_1 - x) - \eta G_2(v_3)
\Big)
f(x) dx
$,
where the PDF is $f(x) = \frac{\eta}{(v_2 - x)^2}$. Thus, we have 
$$
\E[u_1] =
\int_0^{v_2-\eta}
\Big(
v_1 - v_2 + \eta - \eta
\big( 
\frac{v_1 - v_2 + \eta}{v_1 - v_3}
\big)
\Big)f(x) dx 
= 
v_1 - v_2 + \eta 
\Big(
\frac{v_2 - v_3 - \eta}{v_1 - v_3}
\Big).
$$
So we see that in this equilibrium, player $1$ manages to increase her payoff with a unilateral payment. The optimal payment is $\eta^* = (v_2-v_3)/2$.
The payoff for player $2$ is  $\eta$.  

Given the analysis above of Lemmas \ref{thm:lemma-first-price-NE-with-eps-payment} and \ref{thm:lemma-first-price-NE-utilities} with $n$ agents, the extension of Lemma \ref{thm:lemma-first-price-mean-based-NE-CDFs} holds without any change to the proof, noticing the fact that since the minimum of the support is $v_3$, any player $j>2$ never wins the auction, so these agents can have arbitrary distributions with support at most $v_3$ with mean-based dynamics. Since this is a utility-improving deviation for player $1$ from the zero-utility payment profile, the auction is not stable with $n$ players. 

%% file: appendix_standard_definitions.tex
\section{Additional Definitions}\label{sec:appendix-definitions}

For completeness, we provide here several standard definitions.

\vspace{5pt}
\noindent
\emph{Regret minimization:} For a given sequence of action profiles $s^1,\dots,s^T$, the \emph{regret} of agent $i$ is the difference in utility for $i$ between the actual utility in that sequence and the utility of the best fixed action in hindsight:
$
\text{Regret}_i(s^1,\dots,s^T) = \max_{s \in S_i} \sum_{\tau=1}^T u_i(s, s_{-i}^\tau)  - \sum_{\tau=1}^T u_i(s_i^\tau,s_{-i}^\tau).
$ 
A no-regret agent minimizes this quantity in the long term. There are several formulations for this property, including high probability definitions; we give the following definition for completeness.
\begin{definition}
    An algorithm satisfies the (external) regret-minimization property if, for a time horizon parameter $T$ and any $T$-sequence of play of the other players $(s_{-i}^1, \dots,s_{-i}^T)$, where $s_i^t$ and $s_{-i}^t$ denote the actions taken at time $t$ by the algorithm and by the other players, respectively, we have that in expectation over the actions taken by the algorithm, 
    $
    \max_{s \in S_i}  \sum_{\tau=1}^T u_i(s, s_{-i}^\tau)  - \sum_{\tau=1}^T u_i(s_i^\tau,s_{-i}^\tau) = o(T)
    $
    as $T \rightarrow \infty$. An agent is \emph{regret-minimizing} if it satisfies the regret-minimization property.
\end{definition}

\vspace{5pt}
\noindent
\emph{Coarse correlated equilibria:} Coarse correlated equilibria are a weaker notion than correlated equilibria \cite{aumann1974subjectivity}, also known as the \emph{Hannan set} or \emph{Hannan consistent} distributions 
 \cite{hannan1957lapproximation}. See also \cite{roughgarden2015intrinsic}. The simplest definition for our purpose is the following.
\begin{definition}
    A joint distribution of play is a coarse correlated equilibrium if under this distribution all players have in expectation regret at most zero.
\end{definition}

\vspace{5pt}
\noindent
\emph{Mean-based learning algorithms:} \emph{Mean-based learning algorithms} \cite{braverman2018selling} are a family of algorithms that play with high probability actions that are best responses to the history of play. This class was shown to include many standard no-regret algorithms, like \emph{multiplicative weights} (see \cite{arora2012multiplicative} and references therein), \emph{follow the perturbed leader} \cite{hannan1957lapproximation,kalai2005efficient}, and \emph{EXP3} \cite{auer2002nonstochastic}. 
\begin{definition} (From \cite{braverman2018selling}). 
Let {\small$\sigma_{i,t} = \sum_{t'=1}^t u_{i,t'}$}, where {\small$u_{i,t}$} is the utility of action $i$ at time $t$.
An algorithm for the experts problem or the multi-armed bandits problem is $\gamma(T)$-mean-based if it is the case that whenever
{\small$\sigma_{i,t} < \sigma_{j,t} - \gamma(T) T$}, then the probability that the algorithm pulls arm $i$ in round $t$ is at most $\gamma(T)$. An algorithm is mean-based if it is $\gamma(T)$-mean-based for some $\gamma(T) = o(1)$.
\end{definition}